\documentclass[11pt,a4paper]{article}
\usepackage{xcolor}
\usepackage{a4wide}
\usepackage[utf8]{inputenc}
\usepackage[T1]{fontenc}
\usepackage{amssymb}
\usepackage{amsfonts}
\usepackage[fleqn]{amsmath}

\usepackage{amsthm}
\usepackage{mathabx}
\usepackage{stmaryrd}
\usepackage{xspace}
\usepackage[shortlabels]{enumitem}
\setenumerate{itemsep=0.005ex}
\allowdisplaybreaks

\newtheorem{theorem}{Theorem}[section]
\newtheorem{lemma}[theorem]{Lemma}
\newtheorem{corollary}[theorem]{Corollary}

\theoremstyle{definition}
\newtheorem{definition}[theorem]{Definition}

\theoremstyle{remark}

\newcommand{\measure}[1]{\mu\!\left( #1 \right)}

\newcommand{\ceil}[1]{\lceil #1 \rceil }

\newcommand{\tuple}[1]{\langle #1 \rangle}
\newcommand{\fRef}[0]{\text{\sc Ref}}
\newcommand{\lRef}[0]{\text{\sc LimRef}}
\newcommand{\N}[0]{\mathbb{N}}

\newcommand{\bSigma}{\mathbf \Sigma}
\newcommand{\bPi}{\mathbf \Pi}

\renewcommand{\epsilon}{\varepsilon}

\newcommand{\pt}[0]{\Pi^0_3}
\newcommand{\bfPi}[0]{\mathbf \Pi}
\newcommand{\prodd}{{\scriptstyle \prod}}

\DeclareMathOperator{\occ}{occ}

\title{Normal Numbers and  the Borel Hierarchy}

\author{
  \begin{tabular}{ccc}
Ver\'onica Becher& Pablo Ariel Heiber& Theodore A.~Slaman\\
\small Universidad de Buenos Aires & \small Universidad de Buenos Aires  &\small University of California, Berkeley
\\
\small vbecher@dc.uba.ar&\small  pheiber@dc.uba.ar&\small  slaman@math.berkeley.edu
  \end{tabular}
}

\begin{document}
\date{November 1, 2013}
\maketitle

\begin{abstract}
  We show that the set of absolutely normal numbers is $\bPi^0_3$-complete in the Borel hierarchy of
  subsets of real numbers.  Similarly, the set of absolutely normal numbers is $\Pi^0_3$-complete in
  the effective Borel hierarchy.
\end{abstract}
\medskip

\section{Introduction}

What is the descriptive complexity of the set of absolutely normal numbers?  Alexander Kechris posed
this question in the early 1990s when he asked whether the set of real numbers which are normal to
base two is $\bfPi^0_3$-complete in the Borel hierarchy.  Ki and Linton~\cite{KiLin94} proved that,
indeed, the set of numbers that are normal to any fixed base is $\bfPi^0_3$-complete.  However,
their proof technique does not extend to the case of absolute normality, that is, normality to all
bases simultaneously.

We show that the set of absolutely normal numbers  is also $\bPi^0_3$-complete.
In fact, it is $\Pi^0_3$-complete in the effective Borel hierarchy.
We  give, explicitly, a reduction that proves the two  completeness results.
By a feature of this reduction we also provide an alternate proof of  Ki and Linton's theorem.
Our analysis here is a refinement  of our algorithm for computing absolutely normal numbers in~\cite{BecHeiSla1301}.

\section{Preliminaries}

\noindent
{\bf Notation.} 
As usual $\N$ is  the set of positive integers. 
A \emph{base} is an integer $b$ greater than or equal to $2$, 
a {\em digit} in base~$b$ is an element of $\{0,\dots,b-1\}$, 
and a {\em block} in base~$b$ is a finite sequence of digits in base~$b$.  
The length of a block $x$ is $|x|$.
We denote the set of blocks in base $b$ of length $\ell$ by
$\{0,\dots,b-1\}^\ell$.
We write the concatenation of two blocks $x$ and $u$ as  $xu$.
For  arbitrary many blocks $u_i$, for $i:1,2,\ldots,m$,  $\prodd_{1\leq i\leq m}u_i$, is their   concatenation in increasing order of~$i$.
Along the sequel, 
when the starting value for  the index is  $1$, 
we abbreviate this expression 
as $\prodd_{ i\leq m}u_i$.

In case $x$ is a finite or an infinite sequence of digits,
$x \restriction i$ is the subblock of the first $i$ digits of $x$
and  $x[i]$ is the $i$th digit of $x$.
 A digit~$d$ {\em occurs} in~$x$ at position~$i$ if $x[i]=d$.  
 A block~$u$ {\em occurs} in~$x$ at position~$i$ if $x[i+j-1]=u[j]$ for
 $j=1..|u|$.
The number of occurrences of the block $u$ in the block $x$ 
is $\occ(x, u)=\#\{ i: u \mbox{ occurs in } x \mbox{ at position } i \}.$

For each real number $R$ in the unit interval we consider its  unique expansion in base~$b$ 
$R = \sum_{i=1}^{\infty} a_i b^{-i}$, where the $a_i$ are digits in base~$b$, and $a_i<b-1$ infinitely many times.
This last condition over $a_i$ ensures a unique representation for every rational number.
We write $(R)_b$ to denote the expansion of a real $R$ in base~$b$ given by the  sequence $(a_i)_{i\geq 1}$.

\subsection{On normal numbers}

Among the several equivalent definitions of absolute normality the following is the most convenient for our presentation. For a reference see the books \cite{Bug12} or \cite{kuipers}.

\begin{definition} \label{2.1}
\begin{enumerate}
\item A real number~$R$ is \emph{simply normal to base~$b$} if for each digit~$d$ in base~$b$,
    $\lim_{n\to\infty}\occ((R)_b \restriction n, d)/n=1/b.$
\item $R$  is \emph{normal to base $b$}  if it is  simply normal to  the bases $b^\ell$, for every integer  $\ell\geq~1$. 

\item $R$ is \emph{absolutely normal} if it is normal to every base.

\item $R$ is \emph{absolutely abnormal} if it is not normal to any base.
\end{enumerate}
\end{definition}
Notice that absolute normality is equivalent to being simply normal to every base, but
absolute abnormality is not equivalent to being not simply normal to every integer base.

\begin{lemma}\label{2.2}
If a real number $R$ is simply normal to at most finitely many bases then $R$ is absolutely abnormal.
If a real number $R$ is normal to at most finitely many bases then $R$ is absolutely abnormal.
\end{lemma}

\begin{proof}
  Fix a base $b$. By contraposition, if $R$ were normal to base $b$, then it would be simply normal
  to all bases $b^\ell$, a contradiction to the hypothesis of the first claim. Simple normality to all
  bases $b^\ell$ implies by definition normality to all bases $b^\ell$, which contradicts the hypothesis
  of the second claim.
\end{proof}

The {\em simple discrepancy} of a block in base $b$ 
indicates the difference between the actual number of occurrences of the digits in that block  and their expected average.
The definition of normality can be given in terms of discrepancy, see \cite{Bug12}.

\begin{definition}\label{2.3}
  Let $u$ be a block of digits in base $b$.  
The \emph{simple discrepancy} of the block $u$ to base $b$ is
\[
D(u,b)= \max\{ |\occ(u, d)/|u|-1/b\,| : d\in \{0,\dots,b-1\} \}
\]
 Let $\ell$ be a positive integer. 
The  \emph{block discrepancy} of block $u$ to blocks of length $\ell$ in base $b$ is
\[
D_\ell(u,b)=\max\{|\occ(u, v)/|u|-1/b^\ell\,|: v \in \{0,\dots,b-1\}^\ell\},
\] 
\end{definition}
Notice that $D(u,b)$ is a number between $0$ and $1-1/b$, and $D_\ell(u,b)$ is a number between $0$ and $1-1/b^\ell$.
By the definition of simple discrepancy, a real number $R$ is simply normal to base $b$ if and only if
$\lim_{n\to\infty} D((R)_b \restriction n,b) = 0.$ 
Instead of asking for  simple discrepancy for every base $b^\ell$, $\ell\geq 1$,
it is possible to characterize normality using block discrepancy.

\begin{lemma}[Theorem 4.2 in \cite{Bug12}]\label{2.4}
A real number $R$ is normal to base $b$ if and only if for every $\ell\geq 1$, 
$\lim_{n\to\infty} D_\ell((R)_b \restriction n,b) = 0.$
\end{lemma}
The next lemma bounds the discrepancy of a concatenations of blocks. 
We use it very often in the sequel, without making explicit reference to it.

\begin{lemma}
If $u_1,\dots,u_n$ are blocks of digits in base $b$,
\[
D\left(\prodd_{j\le n} u_j, b\right) \le 
\left. {\sum_{j=1}^n D(u_j,b) |u_j| } \middle/  { \sum_{h=1}^n |u_h| }. \right. 
\]
\end{lemma}

\begin{proof}
Let $d$ be a digit in base $b$ which maximizes $\left|\occ\left(\prodd_{j \le n} u_j, d\right)\middle/\left|\prodd_{j \le n} u_j\right| - 1/b\,\right|$.
\begin{align*}
D\left(\prodd_{j \le n} u_j, b\right) 
  =& \left|\occ\left(\prodd_{j \le n} u_j, d\right)\middle/\left|\prodd_{j \le n} u_j\right| - 1/b\,\right| \\
\le& \left| \left(\sum_{j=1}^n \occ(u_j,d) \middle/ \sum_{h=1}^n |u_h|\right) - \frac{1}{b}\,\right| \\
\le& \left. \left| \sum_{j=1}^n \occ(u_j,d) - \frac{|u_j|}{b} \right|\ \middle/\ \sum_{h=1}^n |u_h| \right. \\
\le& \left. \sum_{j=1}^n |u_j| \left|\frac{\occ(u_j,d)}{|u_j|} - \frac{1}{b} \right|\ \middle/\ \sum_{h=1}^n |u_h| \right. \\
\le& \left. \sum_{j=1}^n D(u_j,b) |u_j|\ \middle/\ \sum_{h=1}^n |u_h| \right..
\end{align*}
\end{proof}

Borel's fundamental theorem showing that almost all real numbers are absolutely normal 
is underpinned by the fact that, for any base, almost every
sufficiently long block has small simple discrepancy relative to that base.
We need an explicit bound for the number of blocks of a given length 
having larger simple discrepancy than a given value. 

\begin{lemma}[Lemma 2.5 in \cite{BecHeiSla1301} adapted from Theorem 148 in \cite{HarWri08}]
Let  $p_{b}(k,i)$ be the number of blocks of length $k$ in base $b$ where a given digit occurs
exactly~$i$ times. 
  Fix a base $b$ and a block of length $k$.  For every real $\varepsilon$ such that 
$6/k \leq  \varepsilon \leq 1/b$,  $\sum_{ i=0}^{ k/b-\varepsilon k}    p_b(k,i)$
and  $\sum_{i= k/b+\varepsilon k}^{ k}  p_b(k,i)$ 
are  at most $b^k e^{- b \varepsilon^2 k/6}$.
\end{lemma}

\begin{lemma}[Lemma 2.6 in \cite{BecHeiSla1301}]\label{2.7}
  Let $t \geq 2$ be an integer and let $\epsilon$ and $\delta$ be real numbers between $0$ and $1$, with
  $\epsilon \leq 1/t$.  Let $k$ be the least integer greater than the maximum of $\ceil{6/\epsilon}$ and\linebreak
  $-\ln(\delta/2t) 6/\epsilon^2$.  Then, for all $b\leq t$ and for all
  $k'\geq k$, the fraction of blocks $x$ of length $k'$ in base $b$ for which $D(x,b)>\epsilon$ is
  less than $\delta$.
\end{lemma}

\subsection{On descriptive set theory}

Recall that the Borel hierarchy for subsets of the real numbers is the stratification of the
$\sigma$-algebra generated by the open sets with the usual interval topology.  For references see
Kechris's textbook \cite{Kec95} or Marker's lecture notes \cite{Mar02}.

A set $A$ is $\bSigma^0_1$ if and only if $A$  is open and $A$ is 
$\bPi^0_1$ if and only if $A$ is closed.  
$A$ is $\bSigma^0_{n+1}$ if and only if it is a
countable union of $\bPi^0_{n}$ sets, and  $A$  is $\bPi^0_{n+1}$ if and only if it is a countable
intersection of $\mathbf\Sigma^0_n$ sets. 

$A$ is hard for a Borel class if  and only if every set in the class is reducible to $A$ by a continuous map.
$A$ is complete in a class if it is hard for this class and belongs to the class.

When we restrict to intervals with rational endpoints and computable countable unions and intersections, we obtain 
the  effective or lightface Borel hierarchy. 
One way to present the finite levels of the effective  Borel hierarchy
is by means of the arithmetical hierarchy of formulas in the language of second-order arithmetic. 
Atomic formulas in this language assert  algebraic identities between integers 
or membership of real numbers in intervals with rational endpoints.
A formula in the arithmetic hierarchy involves only quantification over integers.
A formula  is $\Pi^0_0$ and $\Sigma^0_0$ if  all  its quantifiers are bounded.  
It is $\Sigma^0_{n+1}$ if it has the form $\exists x\, \theta$  where $\theta$ is $\Pi^0_n$,
and it  is $\Pi^0_{n+1}$ if it has the form $\forall x\, \theta$ where  $\theta$ is $\Sigma^0_n$.  

A set $A$ of real numbers is $\Sigma^0_n$ (respectively $\Pi^0_n$) in the effective Borel hierarchy if
and only if membership in that set is definable by a formula
which is $\Sigma^0_n$ (respectively
$\Pi^0_n$).
Notice that every $\Sigma^0_n$ set is $\bSigma^0_n$ and every  $\Pi^0_n$ set is $\bPi^0_n$.
In fact for every set $A$ in $\bSigma^0_n$ there is a  $\Sigma^0_n$ formula and a real parameter 
such that membership in $A$ is defined by that $\Sigma^0_n$ formula relative to that real parameter.

$A$ is hard for an  effective Borel class if and only if every set in the class is
reducible to $A$ by a computable map.  
As before, $A$ is complete in an effective  class if it is hard for this class and belongs to the class.
Since computable maps are continuous, 
proofs of hardness in the effective hierarchy often yield proofs of hardness in general by
relativization.  This is the case in our work.

\section{Main Theorem}

By the form of its definition, normality to a fixed base is explicitly a $\pt$ property of real numbers.
The same holds for absolute normality. 
Absolute abnormality is, for all bases, the negation of normality, hence a $\Pi^0_4$ property.

\begin{lemma}[As in \cite{Mar02}]\label{3.1}
The set of real numbers that are normal to a given base is $\Pi^0_3$. \newline
The set of real numbers that are absolutely normal is $\Pi^0_3$. \newline
The set of real numbers that are absolutely abnormal is $\Pi^0_4$.
\end{lemma}

Thus, to prove completeness of the set of absolutely normal real numbers for the class $\Pi^0_3$ we
need only prove hardness.  We prove our hardness result for the Borel hierarchy by relativizing
a hardness result for $\pt$ subsets of the natural numbers.  Let $\cal L$ be the language of
first-order arithmetic.  As usual, a sentence is a formula without free variables.

\begin{theorem}\label{3.2}
There is a computable reduction from $\pt$ sentences $\varphi$ in $\cal L$ to indices
$e$ such that the following implications hold.
\begin{enumerate}
\item[] If $\varphi$ is true then $e$ is the index of a computable absolutely normal number.

\item[] If $\varphi$ is false  then $e$ is the index of a computable absolutely abnormal number.
\end{enumerate}
\end{theorem}
We postpone the  proofs of Theorem~\ref{3.2}  and its corollaries to the end of the next section.

\begin{corollary}\label{3.3}
The set of absolutely normal numbers is $\pt$-complete, and hence    $\mathbf\pt$-complete.
\end{corollary}

The reduction in Theorem~\ref{3.2} gives just two possibilities: absolute
normality or absolute abnormality;
that is, normality to all bases simultaneously, or to no base at all.
Consequently, it also separates normality in base $b$ from non-normality in
base $b$, for any given $b$.

\begin{corollary}\label{3.4}
For every base $b$, the set of normal numbers in base
$b$ is $\pt$-complete, and hence  $\mathbf\pt$-complete.
\end{corollary}
This gives an alternate proof of Ki and Linton's theorem  in \cite{KiLin94} for $\mathbf \pt$-completeness,
that also covers the case of $\pt$-completeness.
Another consequence of Theorem~\ref{3.2} is that absolute abnormality, which
is a $\Pi^0_4$ property, is hard for the classes $\Sigma^0_3$ and   $\bSigma^0_3$.

\section{Proofs}

We shall define a computable reduction that maps $\pt$ sentences  in the language 
of first order arithmetic $\cal L$ 
to indices of  computable infinite sequences of zeros and ones.  
If the given sentence is true then the corresponding
binary sequence is the expansion in base two of an absolutely normal number.
Otherwise, the corresponding binary sequence is the expansion in base two of an
absolutely abnormal number.

Our reduction is the composition of two reductions.  We use Baire space $\N^\N$, 
the set of infinite sequences of positive integers, as an intermediate working space.  
The first reduction  maps sentences from $\cal L$  to programs that produce infinite sequences of positive
integers that reflect the truth or falsity of the given sentences.  
The second reduction maps these programs to ones that produce binary sequences with the appropriate condition on normality. 

By relativizing this reduction, given a $\pt$ formula in second order arithmetic
and a real number $X$, we produce a binary sequence computably from $X$ which is absolutely normal or absolutely abnormal depending on whether the given formula is true of $X$.  This is exactly what is required to establish $\Pi^0_3$-hardness.

\subsection{The first reduction}

Recall that a $\Pi^0_3$ formula in first order arithmetic is equivalent to one starting with a universal
quantifier $\forall$, followed by the quantifier ``there are only finitely many''
$\exists^{<\infty}$ and ended by a computable predicate, see Theorem~XVII and Exercise~14-27 in~\cite{Rog87}.  The computable predicate in this equivalent form comes from the $\Sigma^0_0$ subformula of the original.

\begin{definition}\label{2.1}
The \emph{first reduction} takes sentences in $\cal L$ to positive integers by mapping
$\forall x \exists^{<\infty} y\  C(x,y)$, where $C$ is computable,
to an index of the following program:
For every positive integer $n$ in increasing order, 
let $x = \max\{z\in \N : 2^z \text{ divides } n\}$ and let $y = n/2^x$.
If $y = 1$ or $C(x,y)$ then  append $x,x+1,...,x+y-1$ to the output sequence.
\end{definition}

\begin{lemma}\label{4.2}
  If $\varphi$ is a $\pt$ sentence in $\cal L$ then Definition~\ref{2.1} gives the index of a
  program that outputs an infinite sequence~$f$ of integers such that the subsequence of~$f$'s first
  occurrences is an enumeration of $\N$ in increasing order and the following dichotomy holds:
\begin{enumerate}
\item[]  If $\varphi$ is true then no positive integer occurs infinitely often in $f$.
\item[]  If $\varphi$ is false  then all but finitely many integers occur infinitely often in $f$.
\end{enumerate}
\end{lemma}

\begin{proof}
Assume $\varphi$ of the form  $\forall x \exists^{<\infty} y\  C(x,y)$, where $C$ is computable.
We say that a tuple $\tuple{x,y}$ is \emph{appending} if $y=1$ or $C(x,y)$.
It is clear by inspection that all possible pairs
$\tuple{x,y}$ with $x,y \in \N$ are processed, and that $\tuple{x+1,y}$ and $\tuple{x,y+1}$ are
always processed after $\tuple{x,y}$. Thus, the first occurrence of an integer $x$ in $f$ is due to
the appending tuple $\tuple{x,1}$.  Moreover, if $x > 1$ then the appending tuple $\tuple{x,1}$ is
processed after $\tuple{x-1,1}$. Thus, $x$ occurs for the first time in $f$ after $x-1$. 

Suppose now $\varphi$ is true.
For any $x$, there are finitely many appending tuples of the form $\tuple{x',y}$ with $x' \le
x$. After all such appending tuples have been processed, $x$ will not be appended to the output
sequence.  Thus no positive integer can occur infinitely often in $f$.

Now suppose $\varphi$ is false. Let $x$ be such that there are infinitely many $y$ such that
$C(x,y)$. Let $z$ be any positive integer.  Each time an appending tuple of the form $\tuple{x,y}$
with $z < y$ is processed, $x+z$ is appended to the output.  Since we assumed there are infinitely
many such tuples, $x+z$ is appended to the output an infinite number of times. Thus, all 
integers greater than $x$ occur infinitely often in the output sequence.
\end{proof}

\subsection{The second reduction}

We use the phrase \emph{$b$-adic interval} to refer to a semi-open interval  of the form $[a/{b^m},(a+1)/{b^m})$,
for $a<b^m$.  We move freely between $b$-adic intervals and base-$b$ expansions.  If $x$ is a
base-$b$ block and it is understood that we are working in base $b$, then we let $.x$ denote the
rational number whose expansion in base $b$ has exactly the digits occurring in $x$. 
Given the block $x$, the reals such that their base-$b$ expansions  extend $x$ are exactly those
belonging to the $b$-adic interval $[.x,.x+b^{-|x|})$. 
Conversely, every
$b$-adic interval $[a/{b^m},(a+1)/{b^m})$ corresponds to a block $x$ as above, where $x$ is obtained
by writing $a$ in base $b$ and then prepending a sufficient number of zeros to obtain a block of
length~$m$.
We use $\mu$ to denote Lebesgue measure and  $\log$ to denote logarithm in base~$2$.

Now we  introduce some tools and establish some of their properties.

\begin{lemma}[\cite{BecHeiSla1301}]\label{4.3}
  For every non-empty interval $I$ and base $b$, there is a
  $b$-adic subinterval $I_b$ of $I$ such that $\measure{I_b}\geq \measure{I} /(2b)$.
  Moreover, such subinterval can be computed uniformly from $I$ and $b$.
\end{lemma}

\begin{definition}[\cite{BecHeiSla1301}]\label{4.4}
A \emph{$t$-sequence} is a nested sequence of $t-1$ intervals, $\vec{I}=(I_2,\dots,I_{t})$, such
that $I_2$ is dyadic and for each base $b$, $I_{b+1}$ is a $(b+1)$-adic subinterval of $I_b$
such that $\measure{I_{b+1}}\geq \measure{I_b}/2(b+1)$. 
We let $x_b(\vec{I})$ be the block in base $b$ such that $.x_b(\vec{I})$ is the expansion
of the left endpoint of~$I_b$ in base $b$.
\end{definition}  

We can iteratively apply Lemma~\ref{4.3} to obtain the following corollary.

\begin{corollary}\label{4.5}
  For every non-empty dyadic interval $I$ and every integer $t \ge 2$, there is
  a $t$-sequence that begins with $I$. Moreover, such $t$-sequence can be
  computed uniformly from $I$ and $t$.
\end{corollary}

If $\vec{I} = (I_2,\dots,I_t)$ is a $t$-sequence, then for any  base $b\leq t$ and any real $X\in I_{t}$, $X$ has $x_b(\vec{I})$ as
an initial segment of its expansion in base $b$.  If, further, $\vec{I}'=(I'_2,\dots
I'_{t'})$ is a $t'$-sequence with $t\leq t'$ such that $I'_2\subset I_t$ and $X\in I'_{t'}$, then
for each $b\leq t$, $\vec{I'}$ specifies how to extend $x_b(\vec{I})$ to a
longer initial segment $x_b(\vec{I}')$ of the base $b$ expansion of $X$.
As opposed to arbitrary nested sequences, 
for $t$-sequences there is a function of $t$ that gives a lower bound of the ratio between the measures of $I_b$
and $I_{b'}$, for any two bases $b$ and $b'$ both less than or equal to $ t$.  
That is, assuming $b>b'$, we have 
\[
\measure{I_b} \geq \frac{\measure{I_{b'}}}{2^{b-b'}\; b!/b'!}.
\]
In the sequel we use the inequality above repeatedly.

\newcommand{\kj}[1]{{k_{#1}\ceil{\log({#1}+1)} }}
\newcommand{\dd}[1]{\ceil{-\log(\delta_{#1})}}
\newcommand{\kjd}[1]{\kj{#1} + \dd{#1}}

Our next task is  to define a function that, for an integer $i$ and a $t$-sequence $\vec{I}$ (for some $t$),
constructs an $(i+1)$-sequence inside $\vec{I}$ with  good properties.
The method is to determine the expansion of the rational endpoints of each $b$-adic interval in the  $(i+1)$-sequence.
Since the respective $b$-adic intervals are   nested,  the determination of the expansions  is done by adding  suffixes. 

We introduce three functions of $i$,   $(\delta_i)_{i\geq 1}$, $(k_i)_{i\geq 1}$ and $(\ell_i)_{i\geq 1}$,  
that act as parameters for the construction of an $(i+1)$-sequence.
The integer $k_i$ indicates how many digits in base $(i+1)$ can  be determined in each step;
thus, $\kj{i}$  indicates how many digits in base $2$ can be determined in each step
(keep in mind  that, in general, more digits are needed  to ensure the same precision in base $2$ than in a larger base).
The integer $\ell_i$ limits how many  digits in base $2$ can {\em at~most} be determined in each step.
And the rational $\delta_i$  bounds  the relative measure of any two intervals in two consecutive nested $(i+1)$-sequences.
We call $\fRef_i$ the function that given $\vec{I}$ constructs an $(i+1)$-sequence by recursion. 
It is a search through nested $(i+1)$-sequences until one with good properties is reached. 
The choice we make for 
$(\delta_i)_{i\geq 1}$, $(k_i)_{i\geq 1}$ and $(\ell_i)_{i\geq 1}$,  
allow us to prove the correctness of the construction.

\begin{definition} \label{4.6}
Let   $(k)_{i\geq 1}$ and $(\ell)_{i\geq 1}$ be the computable sequences of positive integers and 
let $(\delta)_{i\geq 1}$ be the computable sequence of positive rational numbers less than $1$  such that,
for each $i\geq 1$,
\begin{align*}
\delta_i       =& \frac{1}{2^{2i-2}\ {(i+1)!}^2} \\
k_i              =& \text{least integer greater than}
\max \left( \left\lceil 6(i+2)\right\rceil , 	 -\ln\left(\frac{\delta_i}{2(i+1)^2}\right)  6 (i+2)^2\right) \\
\ell_i          =& \kjd{i}.
\end{align*}
\end{definition}

\begin{definition}\label{4.7}
The function $\fRef_i$ maps a $(p+1)$-sequence $\vec{I} = (I_2,\dots,I_{p+1})$
into an $(i+1)$-sequence $\fRef_i(\vec{I})$, that we define recursively.
\medskip

\noindent
{\em Initial step $0$.} 
Let $\vec{I}_0=(I_{0,2},\dots,I_{0,i+1})$ be an $(i+1)$-sequence 
where $I_{0,2}$ is the  leftmost dyadic subinterval of $I_{p+1}$ such that $\measure{I_{0,2}} \ge \measure{I_{p+1}}/4$.
\medskip

\noindent
{\em Recursive step $j+1$.} Let  $\vec{I}_{j+1}$ be the $(i+1)$-sequence  such that
\begin{itemize}

\item Let $L$ be the leftmost dyadic subinterval of $I_{j,i+1}$ such that $\measure{L} \ge \measure{I_{j,i+1}}/4$.

\item Partition $L$ into $\kj{i}$ many dyadic subintervals of equal measure
$\displaystyle{2^{-\kj{i}} \measure{L}}$. \\
For each such subinterval $J_2$ of $L$, 
define the $(i+1)$-sequence  $\vec{J}=(J_2,J_3,\dots,J_{i+1})$.

\item\label{4.7} 
Let $\vec{I}_{j+1}$ be the leftmost 
 of the  $(i+1)$-sequences $\vec{J}$ considered above
  such that for each base $b \leq i+1$, $D(u_b(\vec{J}),b) \leq 1/(i+2)$, 
where $u_b(\vec{J})$ is  such that $x_b(\vec{I_j}) u_b(\vec{J}) = x_b(\vec{J})$.
\end{itemize}

\noindent
 Repeat the recursive step until step $n$ when all the following hold for every base $b \leq i+1$:
\[
\begin{array}{llcl}
a. \ & |x_b(\vec{I}_n)|                  &>     & \ell_{i+1}(i+3)  \\
b.\  & D(x_b(\vec{I}_n),b)             &\leq &2/(i+2)   \\
c.\  & D_{2\ell_i}(x_b(\vec{I}_n),b)& >    & b^{-2\ell_i-1}.
 \end{array}
\]

 Finally, let $\fRef_i(\vec{I}) = \vec{I}_n$.
\end{definition}

The following lemmas show that for every positive integer $i$, the function~$\fRef_i$ is well defined and it is computable.

\begin{lemma}
\label{4.8}
There is always a suitable (i+1)-sequence $\vec{J}$ to be selected in the  recursive step of Definition~\ref{4.7}.
\end{lemma}

\newcommand{\SC}[0]{S}
\newcommand{\NC}[0]{N}
\begin{proof}
Consider the recursive step $j+1$ of Definition~\ref{4.7}.
Let ${\SC}$ be the  union of the set of intervals $J_{i+1}$
over the $2^{\kj{i}}$ many $(i+1)$-sequences~$\vec{J}$. 
We have $\measure{L} \ge \measure{I_{j,i+1}}/4$
and, since $\vec{J}$ and $\vec{I_j}$ are $(i+1)$-sequences, we have
\[
\measure{J_{i+1}} \ge \frac{  \measure{J_2} }{2^{i-2} (i+1)!} \quad \text{ and }\quad
\measure{I_{j,i+1}} \ge \frac{  \measure{I_{j,2}} }{2^{i-2} (i+1)!}.
\]
Since the possibilities for $J_2$ form a partition of~$L$,
\[
\measure{\SC} \ge \frac{  \measure{L} }{2^{i-2} (i+1)!} \ge \frac{ \measure{I_{j,i+1}} }{2^i (i+1)!} \ge \frac{  \measure{I_{j,2}} }{2^{2i-2} (i+1)!^2} = \delta_i \measure{I_{j,2}}.
\]
Let us say that an $(i+1)$-sequence $\vec{J}$ is {\em not suitable} if for some  
base $b\leq i+1$,   
\[
D(u_b(\vec{J}),b) > 1/(i+2).
\]  
Let~${\NC}$ be the subset of $\SC$  defined as the union of the set of intervals $J_{i+1}$ 
which occur in $(i+1)$-sequences which are  {\em not suitable}.  
Each $\vec{J}$ considered at stage~$i+1$ 
is such that for every base $b\leq i+1$ each interval  $J_b$ is a subinterval of~$I_{j,b}$.
By definition, $|u_b(\vec{J})|>k_i$ for each $b$ and~$\vec{J}$.
By  Lemma~\ref{2.7} with $t=i+1$, $\epsilon = 1/(i+2)$, $\delta = \delta_i / (i+1)$
and $k=k_i$,
for each base $b\leq i+1$, 
the subset of $I_{j,b}$ consisting of reals with base $b$ expansions starting with
$x_b(\vec{I_j}) u_b(\vec{J})$ for which $D(u_b(\vec{J}),b)>1 / (i+2)$ has measure 
less than  $\delta \measure{I_{j,b}}$,
and hence, less than $\delta \measure{I_{j,2}}$.  
Therefore, 
\[
\displaystyle{\measure{\NC} < (i+1) \delta \measure{I_{j,2}} = \delta_i \measure{I_{j,2}} = \measure{\SC}.}
\]
This proves that ${\SC}$ is a proper superset of ${\NC}$, therefore,
there is a suitable $(i+1)$-sequence.
\end{proof}

\begin{lemma}\label{4.9}
The recursion in Definition~\ref{4.7} finishes for every input sequence $\vec{I}$.
\end{lemma}
\begin{proof}
Using Lemma~\ref{4.3} or Corollary~\ref{4.5} the needed
$b$-adic subintervals with the appropriate measure and $(i+1)$-sequences
can be found computably.
Lemma~\ref{4.8} ensures that a suitable  $\vec{J}$ can always be found in each recursive step.  
All the other tasks in the recursive step are clearly computable.
It remains to check that the ending conditions of the recursion  are eventually~met. \linebreak
Let $u_{b,j+1}$ be such
that 
\[x_b(\vec{I}_j) u_{b,j+1} = x_b(\vec{I}_{j+1})
\]
 and $v_b$ be such that
\[x_b(\vec{I_0}) = x_b(\vec{I}) v_b.
\]
Then, $1 \le |u_{b,j}|$ and
\begin{align*}
|u_{b,j}| =\ & |x_b(\vec{I}_j)| - |x_b(\vec{I}_{j-1})| 
 \\=\  & -\log_b \frac{ \measure{I_{j,b}} }{ \measure{I_{j-1,b}} } 
\\=\ & -\log_b \frac{ \measure{I_{j,b}} }{ \measure{I_{j,2}} }\ \frac{ \measure{I_{j,2}} }{ \measure{I_{j-1,i+1}} }\ \frac{ \measure{I_{j-1,i+1}} }{ \measure{I_{j-1,b}} }
\\\le\   & -\log_b \frac{1}{2^{b-3} b!}\ \frac{1}{4\ 2^{\kj{i}}}\ \frac{1}{2^{i+1-b} (i+1)!/b!} 
\\\le\ & 	-\log \frac{1}{2^{i-2} (i+1)!}\ \frac{1}{4\ 2^{\kj{i}}} 
\\\le\ & \kj{i} -\log \delta_i  
\\\le\ & \ell_i.
\end{align*}
The recursive step establishes $D(u_{b,j}, b) \le 1/(i+2)$, and for any~$k$,
$x_b(\vec{I_k}) = x_b(\vec{I}) v_b \prodd_{j \le k} u_{b,j}.$
Notice that, in each of the three  conditions $(a),(b)$ and~$(c)$, 
the right side of the inequality is fixed.
For condition~$(a)$,
$|x_b(\vec{I_n})| = |x_b(\vec{I}) v_b \prodd_{j \le n} u_{b,j}|$
is strictly increasing on $n$, so it is greater than the required lower bound
for sufficiently large $n$.
For~condition~$(b)$, observe that
\begin{align*}
D(x_b(\vec{I_n}), b) 
 =\ & D(x_b(\vec{I}) v_b \prodd_{j \le n} u_{b,j}, b) \\
\le\ & |x_b(\vec{I}) v_b| / |x_b(\vec{I}_n)| +
	D(\prodd_{j \le n} u_{b,j}, b) \\
\le\ & |x_b(\vec{I}) v_b| / |x_b(\vec{I}_n)| + 1/(i+2).
\end{align*}
In the right hand side, the first term approaches $0$ for large $n$, so the
entire expression is less than $2/(i+2)$ for sufficiently large $n$.
For condition~$(c)$, observe that the recursive step ensures that 
$u_{b,j}$ is never all zeros.  So, a sequence of $2\ell_i$ zeros does not
occur in $\prodd_{j \le n} u_{b,j}$.
By definition, 
\[
D_{2\ell_i}(x_b(\vec{I}) v_b \prodd_{j \le n} u_{b,j}, b) \ \ge \
\left|\frac{ \occ(x_b(\vec{I}) v_b \prodd_{j \le n} u_{b,j}, 0^{2\ell_i}) }{ |x_b(\vec{I}) v_b \prodd_{j \le n} u_{b,j}| } - \frac{1}{b^{2\ell_i}} \right|.
\]
Since $\occ(x_b(\vec{I}) v_b \prodd_{j \le n} u_{b,j}, 0^{2\ell_i})$
is bounded by a constant,
for sufficiently large $n$, the discrepancy
$D_{2\ell_i}(x_b(\vec{I}) v_b \prodd_{j \le n} u_{b,j}, b)$
is arbitrarily close to $b^{-2\ell_i}$.
\end{proof}

\begin{lemma}  \label{4.10}
Let $\vec{I} $
be an arbitrary $(p+1)$-sequence,
$i \ge 1$ be an integer and $\vec{R}$ be 
$\fRef_i(\vec{I})$.
For every base $b \le \min(i,p)+1$,
\begin{enumerate}
\item $R_2 \subseteq I_{p+1}$
\item $D(x_b(\vec{R}),b) \le 2/(i+2)$
\item $D_{2\ell_i}(x_b(\vec{R}), b) > b^{-2\ell_i-1}$
\item $|x_b(\vec{R})| > \ell_{i+1} (i+3)$
\item For each $\ell$ such that $|x_b(\vec{I})| \le \ell \le |x_b(\vec{R})|$, 
\end{enumerate}
\[
D(x_b(\vec{R}) \restriction \ell,b) \le D(x_b(\vec{I}), b) + \dd{p}/|x_b(\vec{I})| + 1/(i+2) + \ell_i/|x_b(\vec{I})|.
\]
\end{lemma}

\begin{proof}
As in the proof of Lemma~\ref{4.9}, let $u_{b,j+1}$ be such
that $x_b(\vec{I}_j) u_{b,j+1} = x_b(\vec{I}_{j+1})$ and  let $v_b$ be such that
$x_b(\vec{I_0}) = x_b(\vec{I}) v_b$.
Then, $1 \le |u_{b,j}| \le \ell_i$, $D(u_{b,j}, b) \le 1/(i+2)$, and for any $k$ it holds that
$x_b(\vec{I_k}) = x_b(\vec{I}) v_b \prodd_{j \le k} u_{b,j}.$

Fix a base  $b$. 
Point (1) follows by induction in the recursive steps in the definition of $\fRef_i(\vec{I})$,
since each subsequent interval is contained in the previous one.
Points (2), (3) and (4) follow from the termination condition in that recursion.
For point (5), use the above definition of $v_b$ and the parameter $\delta_p$
(see Definition \ref{4.6}).
\begin{align*}
|v_b| =\ & |x_b(\vec{I_0})| - |x_b(\vec{I})| \\
=\ & -\log_b \frac{ \measure{I_{0,b}} }{ \measure{I_b} } \\
=\ & -\log_b \frac{ \measure{I_{0,b}} }{\measure{I_{0,2}}} \frac{\measure{I_{0,2}}}{\measure{I_{p+1}}} \frac{\measure{I_{p+1}}}{ \measure{I_b} } \\
=\ & -\log_b \frac{ 1 }{ 2^{b-3} b! } \frac{1}{4} \frac{1}{ 2^{p+1-b} (p+1)!/b! } \\
=\ & -\log \frac{ 1 }{ 2^p (p+1)! } \\
\le\ & -\log \delta_p \\
\le\ & \dd{p}. 
\end{align*}
Then, for each $m$,
$D(\prodd_{j \le m} u_{b,j}, b) \le 1/(i+2)$ and
\[
D(x_b(\vec{I}) v_b \prodd_{j \le m} u_{b,j}, b)\ \le \ D(x_b(\vec{I}), b) + \dd{p}/|x_b(\vec{I})| + 1/(i+2).
\]
Finally, fix $\ell$ and let $m$ and $\ell'$ be such that
$(x_b(\vec{I}) v_b \prodd_{j \le m} u_{b,j}) (u_{b,m+1} \restriction \ell') = 
x_b(\vec{R}) \restriction \ell$. Then, 
\begin{align*}
D(x_b(\vec{R}) \restriction \ell, b) 
=\ & D((x_b(\vec{I}) v_b \prodd_{j \le m} u_{b,j}) (u_{b,m+1} \restriction \ell'), b) \\
\le\ & D(x_b(\vec{I}), b) + \dd{p}/|x_b(\vec{I})| + 1/(i+2) + |u_{b,m+1}|/|x_b(\vec{I})| \\
\le\ & D(x_b(\vec{I}), b) + \dd{p}/|x_b(\vec{I})| + 1/(i+2) + \ell_i/|x_b(\vec{I})|.
\end{align*}
\end{proof}

\begin{definition}\label{4.11}
 We define the function $\lRef$ that takes infinite sequences of positive integers 
to real numbers in the unit interval, and
\[
\lRef(f) \text{ is the unique element in } \bigcap_{j=1}^\infty (\vec{R}_j)_2,
\text{ where } 
\vec{R}_0 = ([0,1))  \text{ and } \vec{R}_{j+1} = \fRef_{f_{j+1}}(\vec{R}_j).
\]
\end{definition}

That is, $\lRef(f)$ is the real obtained by iterating applications of $\fRef_{i}$ 
where $i$ is determined by the positive integers in $f$. 
By point (1) of Lemma~\ref{4.10}, for each $j\geq 1$,
$\lRef(f)$ is inside every interval in every $(f_j+1)$-sequence $\vec{R}_j$, and
therefore, for each base $b \le f_j+1$, $x_b(\vec{R}_j)$ is a prefix of
$(\lRef(f))_b$.

\begin{lemma}\label{4.12}
Let $f$ be a sequence of positive integers such that the subsequence of $f$'s 
first occurrences  is an enumeration of $\N$ in increasing order and no positive integer
occurs infinitely often in~$f$.
Then, $\lRef(f)$ is an absolutely normal number.
\end{lemma}

\begin{proof}
Fix a base $b$ and $\epsilon > 0$. We prove that
$D((\lRef(f))_b \restriction \ell,b) \le \epsilon$ for each sufficiently large
$\ell$.
Let $j_0$ be large enough such that $f_j > \max(b,\ceil{8/\epsilon})$ for every
$j \ge j_0$. Consider $\ell > |x_b(\vec{R}_{j_0})|$, and noticing that
$(|x_b(\vec{R}_j)|)_{j \in \N}$ is an increasing sequence, let $j$ be such that
$|x_b(\vec{R}_j)| \le \ell < |x_b(\vec{R}_{j+1})|$.
Observe that
$(\lRef(f))_b \restriction \ell = x_b(\vec{R}_{j+1}) \restriction \ell$.
Now note that $1/(f_j+2) \le \epsilon/8$ and apply point~(2) of Lemma~\ref{4.10}
to $\vec{R}_j = \fRef_{f_j}(\vec{R}_{j-1})$ to conclude that
\[
D(x_b(\vec{R}_j), b) \le 2/(f_j+2) \le \epsilon/4.
\]
By hypothesis, $f_j,f_{j+1} > b > 1$, so
let $j_1 < j$ be such that
$f_{j_1} = f_j-1$ and $j_2 < j+1$ be such that $f_{j_2} = f_{j+1}-1$.
By~point (4) of Lemma~\ref{4.10}, 
\[
|x_b(\vec{R}_j)| \ge |x_b(\vec{R}_{j_1})| > \ell_{f_j}(f_j+2) > \dd{f_j}(f_j+2),
\]
then,
\[
\dd{f_j}/|x_b(\vec{R}_j)| < 1/(f_j+2) \le \epsilon/8.
\] 
Similarly,
\[
|x_b(\vec{R}_j)| \ge |x_b(\vec{R}_{j_2})| > \ell_{f_{j+1}}(f_{j+1}+2),
\]
then,
\[
\ell_{f_{j+1}}/|x_b(\vec{R}_j)| < 1/(f_{j+1}+2) \le \epsilon/8.
\]
Now consider Lemma~\ref{4.10} again, but applied to
$\vec{R}_{j+1} = \fRef_{f_{j+1}}(\vec{R}_j)$. By point (5),
\[
D(x_b(\vec{R}_{j+1}) \restriction \ell, b) \le D(x_b(\vec{R}_j), b) \ +\  \dd{f_j}/|x_b(\vec{R}_j)| \ +\  1/(f_j+2) \ +\  \ell_{f_{j+1}}/|x_b(\vec{R}_j)|.
\]
By the bounds established above, each term on the right part of the inequality
is at most~$\epsilon/4$. So,
\[
D((\lRef(f))_b \restriction \ell, b) = D(x_b(\vec{R}_{j+1}) \restriction \ell,b) \le \epsilon.
\]
By the choice of $j_0,j,j_1,j_2$ the sequences
$\vec{R}_{j_0}, \vec{R}_j, \vec{R}_{j_1}, \vec{R}_{j_2}$ contain a $b$-adic
interval, hence the function $x_b$ is defined on them.
\end{proof}

\begin{lemma}\label{4.13}
Let $f$ be a sequence of positive integers such that all but finitely many occur infinitely often in $f$.
Then, $\lRef(f)$ is absolutely abnormal.
\end{lemma}

\begin{proof}
Fix  a base $b$ such that $b$ appears infinitely often in~$f$.
By the conditions imposed on~$f$, $(\lRef(f))_{b+1}$ has infinitely many prefixes of
the form $x_{b+1}(\fRef_{b}(\vec{I}))$ for
some $\vec{I}$.
By point~(3) of Lemma~\ref{4.10},
\[
D_{2\ell_{b}}(x_{b+1}(\fRef_{b}(\vec{I})),b+1) > (b+1)^{-2\ell_{b}-1}.
\]
Hence, for  infinitely many prefixes of  $(\lRef(f))_{b+1}$
their discrepancy to blocks of length $2\ell_{b}$ in base $b+1$
 is bounded away from $0$.
Then, by Lemma~\ref{2.4}, $\lRef(f)$ is not normal to base~$b+1$. 
Since all but finitely many bases can be chosen as $b+1$, $\lRef(f)$ is not normal
to all but finitely many bases. By Lemma~\ref{2.2} it is absolutely abnormal.
\end{proof}

We are ready to define the second reduction.

\begin{definition}\label{def:second}
The \emph{second reduction}  maps the index for a computable  infinite sequence of integers $f$
to the  index  for  the  infinite binary sequence $(\lRef(f))_2$.
\end{definition}

Since $\lRef(f)$ is uniformly computable from the input $f$, the second reduction is computable.

\subsection{Proof of the main theorem and corollaries}

\begin{proof}[Proof of Theorem~\ref{3.2}]
The needed reduction is the composition of the first reduction, given in Definition~\ref{2.1}, 
and the second reduction, given in Definition~\ref{def:second}. 
Apply Lemma~\ref{4.2}  for the first reduction
and Lemma~\ref{4.12} for the second reduction to obtain the first implication 
in Theorem~\ref{3.2}.
Apply Lemma~\ref{4.2}  for the first reduction
and  Lemma~\ref{4.13} for the second reduction to obtain the second implication.
\end{proof}

\begin{proof}[Proof of Corollary~\ref{3.3}]
Lemma~\ref{3.1} states that the corresponding sets are in the $\pt$ and $\mathbf \pt$ classes.
The hardness result  in the effective case is immediate from Theorem~\ref{3.2} by relativization.   We have the reduction from a $\pt$ sentence in first order arithmetic to an appropriate index for a computable real number.  By relativization, we obtain a reduction from a $\pt$ statement about a real number $X$ to an appropriate index of a real number which is computable from $X$.

For the general case, recall that to  prove hardness of subsets of reals  at levels in
the Borel hierarchy it is sufficient to consider subsets of Baire space $\N^\N$, because there is a
continuous function from the real numbers to $\N^\N$ that preserves  $\bPi^0_3$ definability.  
Baire space admits a syntactic representation of the levels in the Borel hierarchy in arithmetical terms, 
namely a subset of $\N^\N$ can be defined by a $\pt$ formula with a fixed parameter $P\in\N^\N$.
The analysis  given for the effective case, but now relativized to $X$ and $P$, applies.
\end{proof}

\begin{proof}[Proof of Corollary~\ref{3.4}]
Observe that  normality in all bases implies normality in each base.
And  absolute abnormality is lack of  normality in every base.
Thus, the same reductions used in the proof of  Corollary~\ref{3.3} also prove the 
completeness results for just one fixed base.
\end{proof}
\bigskip
\bigskip

\noindent {\bf Acknowledgments.} This research received support from  CONICET and  Agencia Nacional de Promoci\'{o}n Cient\'{i}fica y Tecnol\'{o}gica, Argentina,  from the National Science Foundation, USA, under Grant No. DMS-1001551, 
and from the Simons Foundation.
V. Becher and P. A. Heiber are members of the Laboratoire International Associ\'e INFINIS, 
\mbox{CONICET}/Universidad de Buenos Aires – CNRS/Universit\'e Paris Diderot. 
This work was done while the authors participated in the Buenos Aires Semester in Computability,
Complexity and Randomness,~2013.

\bibliography{pi3}

\begin{thebibliography}{1}

\bibitem{BecHeiSla1301}
Ver{{\'o}}nica Becher, Pablo~Ariel Heiber, and Theodore~A. Slaman.
\newblock A polynomial-time algorithm for computing absolutely normal numbers.
\newblock {\em Information and Computation}, 232:1--9, 2013.

\bibitem{Bug12}
Yann Bugeaud.
\newblock {\em Distribution Modulo One and Diophantine Approximation}.
\newblock Number 193 in Cambridge Tracts in Mathematics. Cambridge University
  Press, Cambridge, UK, 2012.

\bibitem{HarWri08}
G.~H. Hardy and E.~M. Wright.
\newblock {\em An introduction to the theory of numbers}.
\newblock Oxford University Press, Oxford, sixth edition, 2008.

\bibitem{Kec95}
Alexander~S. Kechris.
\newblock {\em Classical descriptive set theory}, volume 156 of {\em Graduate
  Texts in Mathematics}.
\newblock Springer-Verlag, New York, 1995.

\bibitem{KiLin94}
Haseo Ki and Tom Linton.
\newblock Normal numbers and subsets of {$\mathbb{N}$} with given densities.
\newblock {\em Fundamenta Mathematica}, 144(2):163--179, 1994.

\bibitem{kuipers}
L.~Kuipers and H.~Niederreiter.
\newblock {\em Uniform distribution of sequences}.
\newblock Dover, 2006.

\bibitem{Mar02}
David Marker.
\newblock Descriptive set theory.
\newblock Course Notes, 2002.

\bibitem{Rog87}
Hartley Rogers, Jr.
\newblock {\em Theory of recursive functions and effective computability}.
\newblock MIT Press, Cambridge, MA, second edition, 1987.

\end{thebibliography}

\end{document}